\documentclass[a4paper]{article}

\usepackage{amssymb}
\usepackage{amsmath}
\usepackage{amsfonts, amssymb}
\usepackage{amscd}

\usepackage{amssymb}
\usepackage{amsthm}
\usepackage{amsmath}
\usepackage[utf8]{inputenc}
\usepackage{hyperref}
\usepackage[pdftex]{graphicx}
\usepackage{bold-extra}
\usepackage{tocbibind}
\usepackage{pgf}
\usepackage{pgfgantt}

\usepackage{amsfonts}
\usepackage{latexsym}
\usepackage{epsfig}
\usepackage{graphicx}
\usepackage{url}

\newcommand{\PAut}{{\operatorname{PAut}}}
\newcommand{\Sym}{{\operatorname{Sym}}}
\newcommand{\wt}{{\operatorname{wt}}}
\newcommand{\PD}{{\operatorname{PD}}}
\newcommand{\AGL}{{\operatorname{AGL}}}
\newcommand{\GL}{{\operatorname{GL}}}

\newcommand{\Z}{{\mathbb{Z}}}

\newcommand{\R}{{\mathcal R}}
\newcommand{\T}{{\mathcal T}}

\newtheorem{definition}{Definition}
\newtheorem{proposition}[definition]{Proposition}
\newtheorem{theorem}[definition]{Theorem}
\newtheorem{example}[definition]{Example}
\newtheorem{lemma}[definition]{Lemma}
\newtheorem{corollary}[definition]{Corollary}

\title{Partial permutation decoding for binary linear and $Z_4$-linear Hadamard codes \thanks{This work was partially supported by the
Spanish MINECO under Grant TIN2013-40524-P, and by the Catalan AGAUR under Grant 2014SGR-691. The material in this paper was presented in part at IX ``Jornadas de Matem\'{a}tica Discreta y Algor\'{\i}tmica'' in Tarragona, Spain, 2014 \cite{BV}.}}

\author{Roland D. Barrolleta\thanks{Departament d'Enginyeria de la Informaci\'{o} i de les Comunicacions,
Universitat Aut\`{o}noma de Barcelona, e-mails: rolanddavid.barrolleta@uab.cat and merce.villanueva@uab.cat.} \and
 Merc\`{e} Villanueva\footnotemark[2] } 

\providecommand{\keywords}[1]{\textbf{\textit{Index terms---}} #1}

\begin{document}

\maketitle

\begin{abstract}
Permutation decoding is a technique which involves finding a subset $S$, called $\PD$-set, of the permutation automorphism group of a code $C$ in order to assist in decoding. An explicit construction of $\left \lfloor{\frac{2^m-m-1}{1+m}} \right \rfloor$-$\PD$-sets of minimum size $\left \lfloor{\frac{2^m-m-1}{1+m}} \right \rfloor + 1$ for partial permutation decoding for binary linear Hadamard codes $H_m$ of length $2^m$, for all $m\geq 4$,
is described. 
Moreover, a recursive construction to obtain $s$-$\PD$-sets of size $l$ for $H_{m+1}$ of length $2^{m+1}$, from a given $s$-$\PD$-set of the same size for $H_m$, is also established. These results are generalized to find $s$-$\PD$-sets for (nonlinear) binary Hadamard codes of length $2^m$,
called $\Z_4$-linear Hadamard codes, which are obtained as the Gray map image of quaternary linear codes of length $2^{m-1}$.
\end{abstract}


\section{Introduction}

Denote by $\mathbb Z_2$ and $\mathbb Z_4$ the rings of integers modulo 2 and modulo 4, respectively. Let $\Z_2^n$ denote the set of all binary vectors of length $n$ and let $\mathbb Z_4^n$ be the set of all $n$-tuples over the ring $\mathbb Z_4$.  The \textit{Hamming weight} $\wt(v)$ of a vector $v \in \mathbb Z_2^n$ is the number of nonzero coordinates in $v$. The \textit{Hamming distance} $d(u,v)$ between two vectors $u, v \in \mathbb Z_2^n$ is the number of coordinates in which $u$ and $v$ differ, that is, $d(u,v)=\wt(u+v)$. Let $e_i$ be the binary vector or tuple over $\Z_4$ with a one in the $i$th coordinate and zeros elsewhere. Let $\mathbf{0}, \mathbf{1}, \mathbf{2}$ and $\mathbf {3}$ be the binary vectors or tuples over $\Z_4$ having 0, 1, 2 and 3, respectively, repeated in each coordinate. It will be clear by the context whether we refer to binary vectors or tuples over $\Z_4$.

Any nonempty subset $C$ of $\mathbb Z_2^n$ is a binary code and a subgroup of $\mathbb Z_2^n$ is called a {\em binary linear code}. Equivalently, any nonempty subset $\mathcal C$ of $\mathbb Z_4^n$ is a quaternary code and a subgroup of $\mathbb Z_4^n$ is called a {\em quaternary linear code}. Quaternary codes can be seen as binary codes under the usual Gray map $\Phi: \mathbb Z_4^n \rightarrow \mathbb Z_2^{2n}$ defined as $\Phi((y_1,\ldots,y_n))=(\phi(y_1),\ldots, \phi(y_n))$, where $\phi(0)=(0,0)$, $\phi(1)=(0,1)$, $\phi(2)=(1,1)$, $ \phi(3)=(1,0)$,
for all $y=(y_1, \dots, y_{n}) \in \mathbb Z_4^{n}$. If $\mathcal C$ is a quaternary linear code, the binary code $C=\Phi(\mathcal C)$ is said to be a  $\mathbb Z_4$-{\em linear code}. Moreover, since $\mathcal C$ is a subgroup of $\mathbb Z_4^n$, it is isomorphic to an abelian group $\mathbb Z_2^{\gamma}\times \mathbb Z_4^{\delta}$ and we say that $\mathcal C$ (or equivalently the corresponding $\mathbb Z_4$-linear code $C=\Phi(\mathcal C)$) is of type $2^{\gamma}4^{\delta}$ \cite{Z4}.

Let $C$ be a binary code of length $n$ and size $|C|=2^k$. For a vector $v \in \mathbb Z_2^n$ and a set $I \subseteq\{1, \dots, n\}$, we denote by $v_{I}$ the restriction of $v$ to the coordinates in $I$ and by $C_I$ the set $\{v_I : v \in C\}$. A set $I\subseteq \{1, \dots, n\}$ of $k$ coordinate positions is an {\em information set} for $C$ if $|C_I|=2^k$. If such an $I$ exists, $C$ is said to be a {\em systematic code}. For each information set $I$ of size $k$, the set $\{1, \dots, n\} \backslash I$ of the remaining $n-k$ coordinate positions is a {\em check set} for $C$. 


Let $\operatorname{Sym}(n)$ be the symmetric group of permutations on the set $\{1,\dots,n\}$ and let $\operatorname{id}\in \operatorname{Sym}(n)$ be the identity permutation. The group operation in $\operatorname{Sym}(n)$ is the function composition $\sigma_1\sigma_2$, which maps any element $x$ to $\sigma_1(\sigma_2(x))$, $\sigma_1, \sigma_2 \in \Sym(n)$. A $\sigma \in \operatorname{Sym}(n)$ acts linearly on words of $\mathbb Z_2^n$ or  $\mathbb Z_4^n$ by permuting their coordinates as follows: $\sigma((v_1, \dots, v_n))=(v_{\sigma^{-1}(1)}, \dots, v_{\sigma^{-1}(n)})$. The {\em permutation automorphism group} of $\mathcal C$ or $C=\Phi(\mathcal C)$, denoted by $\operatorname{PAut}({\mathcal C})$ or $\operatorname{PAut}(C)$, respectively, is the group generated by all permutations that preserve the set of codewords.

A {\em binary Hadamard code} of length $n$ has $2n$ codewords and minimum distance $n/2$. It is well-known that there exists an unique binary linear Hadamard code $H_m$ of length $n=2^m$, for any $m \geq 2$.
The quaternary linear codes such that, under the Gray map, give a binary Hadamard code are called {\em quaternary linear Hadamard codes} and the corresponding $\mathbb Z_4$-linear codes are called $\Z_4$-{\em linear  Hadamard codes}. These codes have been studied and classified in \cite{K, PRV}, and their permutation automorphism groups have been determined in \cite{KVi, PPV}.

Permutation decoding is a technique, introduced in \cite{M} by MacWilliams for linear codes, which involves finding a subset of the permutation automorphism group of a code in order to assist in decoding. A new permutation decoding method for $\Z_4$-linear codes (not necessarily linear) was introduced in \cite{BeBoFeVi}. In general, the method works as follows. Given a systematic $t$-error-correcting code $C$ with information set $I$, we denote by $y=x+e$ the received vector, where $x \in C$ and $e$ is the error vector. Suppose that at most $t$ errors occur, that is, $\wt(e)\leq t$. The permutation decoding consists on moving all errors in $y$ out of
$I$, by using an automorphism of $C$. This technique is strongly based on the existence of some special subsets of $\mathrm{PAut}(C)$, called $\PD$-sets.
Specifically, a subset $S \subseteq \operatorname{PAut}(C)$ is said to be an {\em $s$-$\PD$-set} for the code $C$ if every $s$-set of coordinate positions is moved out of $I$ by at least one element of $S$, where $1 \leq s \leq t$. When $s=t$, $S$ is said to be a {\em $\PD$-set}.

In \cite{FKM}, it is shown how to find $s$-$\PD$-sets of size $s+1$ that satisfy the Gordon-Sch\"onheim bound for partial permutation decoding for the binary simplex code of lenght $2^{m}-1$ for all $m\geq 4$ and $1<s\leq \left \lfloor{\frac{2^m-m-1}{m}} \right \rfloor$. In this paper, following the same technique, similar results for binary linear and $\Z_4$-linear Hadamard codes are established. In \cite{S}, 2-$\PD$-sets of size 5 and 4-$\PD$-sets of size $ m+1 \choose 2$ + 2 are found for binary linear Hadamard codes $H_m$, for all $m>4$. Small $\PD$-sets that satisfy the Gordon-Sch\"onheim bound have also been found for binary Golay codes \cite{G,W} and for the binary simplex code $S_4$ \cite{KV2}.

This work is organized as follows.
In Section \ref{sec:size}, we prove that the Gordon-Sch\"onheim bound
can be adapted to systematic codes, not necessarily linear. Furthermore, we apply this bound on the minimum size of $s$-$\PD$-sets to binary linear and $\Z_4$-linear Hadamard codes, which are systematic but nonlinear in general, and we prove that their minimum size is $s+1$. In Section \ref{sec:binarylinear}, we regard the permutation automorphism group $\PAut(H_m)$ as a certain subgroup of the general linear group $\operatorname{GL}(m+1, 2)$ and we provide a criterion on subsets of matrices of such subgroup to be an $s$-$\PD$-set of size $s+1$ for $H_m$.
In Section \ref{sec:binarylinearRecursive}, we define recursive constructions to obtain $s$-$\PD$-sets of size $l$ for $H_{m+1}$ from a given $s$-$\PD$-set of the same size for $H_m$, where  $l\geq s+1$. Finally, in Sections \ref{sec:Z4linear} and \ref{sec:Z4linearRecursive}, we establish equivalent results for (nonlinear) $\Z_4$-linear Hadamard codes.

\section{Minimum size of $s$-$\PD$-sets for Hadamard codes}
\label{sec:size}

There is a well-known bound on the minimum size of $\PD$-sets for linear codes based on the length, dimension and minimum distance of such codes that can be adapted to systematic codes (not necessarily linear) easily.

\begin{proposition}
\label{boundsize}
Let $C$ be a systematic $t$-error correcting code of length $n$, size $|C|=2^k$ and minimum distance $d$. Let $r=n-k$ be the redundancy of $C$. If $S$ is a $\PD$-set for $C$, then
\begin{equation} \label{boundsizeIneq} |S| \geq \left \lceil{\frac{n}{r} \left \lceil{\frac{n-1}{r-1}\left \lceil{\dots \left \lceil{\frac{n-t+1}{r-t+1} }\right \rceil \ \dots}\right \rceil} \right \rceil} \right \rceil.
\end{equation}
\end{proposition}

The above inequality (\ref{boundsizeIneq}) is often called the {\em Gordon-Sch\"onheim bound}. The
result given by Proposition~\ref{boundsize} is quoted and proved for linear codes in \cite{H}. We can follow the same proof, since the linearity of the code is only used to guarantee that the code is systematic. In \cite{BeBoFeVi}, it is shown that $\mathbb Z_4$-linear codes are systematic, and a systematic encoding is given for these codes. Therefore, the result can be applied to any $\Z_4$-linear code, not necessarily linear.

The Gordon-Sch\"onheim bound can be adapted to $s$-$\PD$-sets for all $s$ up to the error correcting capability of the code. Note that the error-correcting capability of any binary linear or $\Z_4$-linear Hadamard code  of length $n=2^m$ is $t_m=\left \lfloor{(d-1)/2} \right \rfloor=\left \lfloor{(2^{m-1}-1)/2} \right \rfloor =2^{m-2}-1$ \cite{MSl}. Moreover, all these codes are systematic and have size $2n=2^{m+1}$. Therefore, the right side of the bound given by (\ref{boundsizeIneq}), for binary linear and $\mathbb Z_4$-linear Hadamard codes of length $2^m$ and for all $1\leq s \leq t_m$, becomes
\begin{equation} \label{function_gm}
g_m(s)=\left \lceil{\frac{2^m}{2^m-m-1}\left \lceil{\frac{2^m-1}{2^m-m-2}\left \lceil{\dots \left \lceil{\frac{2^m-s+1}{2^m-m-s}} \right \rceil } \right \rceil \dots } \right \rceil} \right \rceil .
\end{equation}
We compute the minimum value of $g_m(s)$ in the following lemma.
\begin{lemma}
\label{L1}
Let $m$ be an integer, $m\geq 4$. For $1\leq s \leq t_m$,
\begin{equation*}g_m(s)=\left \lceil{\frac{2^m}{2^m-m-1}\left \lceil{\frac{2^m-1}{2^m-m-2}\left \lceil{\dots \left \lceil{\frac{2^m-s+1}{2^m-m-s}} \right \rceil } \right \rceil \dots } \right \rceil} \right \rceil  \geq s+1,
\end{equation*}
where $t_m=2^{m-2}-1$ is the error-correcting capability of any binary linear and $\Z_4$-linear Hadamard code of length $2^m$.
\end{lemma}

\begin{proof}
We need to prove that $g_m(s) \geq s+1$. This fact is clear, since the central term
\begin{equation*}
\displaystyle\left \lceil{\frac{2^m-s+1}{2^m-m-s}} \right \rceil =2
\end{equation*}
for all $s \in \{1, \dots, 2^{m-2}-1\}$, and in each stage of the ceiling function working from inside, $g_m(s)$ increases its value by at least 1. 
\end{proof}

The smaller the size of the $\PD$-set is, the more efficient permutation decoding becomes. Because of this, we will focus on the case when we have that $g_m(s)=s+1$. For each binary linear and $\Z_4$-linear Hadamard code of length $2^m$, $m\geq 4$, we define the following integer:
\begin{equation*}
f_{m}=\operatorname{max}\{s \; : \; 2 \leq s, \; g_m(s)=s+1\},
\end{equation*}
which represents the greater $s$ in which we can find $s$-$\PD$-sets of size $s+1$. The following result characterize this parameter from the value of $m$. Note that for $m=3$, since the error-correcting capability is $t_3=1$, the permutation decoding becomes unnecessary and we do not take it into account in the results.

\begin{lemma}
\label{bound}
Let $m$ be an integer, $m\geq 4$. Then, $f_{m}=\left \lfloor{\frac{2^m-m-1}{1+m}} \right \rfloor.$
\end{lemma}

\begin{proof}
The result is easy to prove by Lemma~\ref{L1} and following a similar argument as the one in the proof of Lemma 2 in \cite{FKM}. 
\end{proof}

\section{Finding $s$-$\PD$-sets of size $s+1$ for binary linear Hadamard codes}
\label{sec:binarylinear}

For any $m\geq 2$, there is an unique binary linear Hadamard code $H_m$ of length $2^m$ \cite{MSl}.
A generator matrix $G_{m}$ for $H_m$ can be constructed as follows:
\begin{equation}
\label{GH}
G_{m}= \left(\begin{array}{cc}
1 & \bf{1}    \\
\bf{0} & G'  \\
\end{array} \right),
 \end{equation}
where $G'$ is any matrix having as column vectors the $2^m-1$ nonzero vectors from $\mathbb Z_2^{m}$, with the vectors $e_i$, $i \in \{1, \dots, m\}$, in the first $m$ positions. Note that $G'$ can be seen as a generator matrix of the binary simplex code of length $2^m-1$.

By construction, from~(\ref{GH}), it is clear that $I_m=\{1,\ldots,m+1\}$ is an information set for  $H_m$. Let $w_i$ be the $i$th column vector of $G_{m}$, $i\in \{1,\ldots,2^m\}$. By labelling the coordinate positions with the columns of $G_{m}$, we can take as an information set $I_m$ for $H_m$ the first $m+1$ column vectors of $G_{m}$ considered as row vectors, that is, $I_m= \{w_1, \dots, w_{m+1}\}=\{e_1,e_1+e_2, \ldots, e_1+e_{m+1}\}$.
Then, depending on the context, $I_{m}$ will be taken as a subset of $\{1, \dots, 2^{m}\}$ or as a subset of $\{1\} \times \mathbb Z_2 ^{m}$.
\begin{example}
Let $H_4$ be the binary linear Hadamard code of length 16 with generator matrix
\begin{equation} \label{GH4}
G_{4}=\left(\begin{array}{cccccccccccccccc}
1 & 1 & 1 & 1 & 1 & 1 & 1 & 1 & 1 & 1 & 1 & 1 & 1 & 1 & 1 & 1 \\
0 & 1 & 0 & 0 & 0 & 1 & 0 & 0 & 1 & 1 & 0 & 1 & 0 & 1 & 1 & 1 \\
0 & 0 & 1 & 0 & 0 & 1 & 1 & 0 & 1 & 0 & 1 & 1 & 1 & 1 & 0 & 0 \\
0 & 0 & 0 & 1 & 0 & 0 & 1 & 1 & 0 & 1 & 0 & 1 & 1 & 1 & 1 & 0 \\
0 & 0 & 0 & 0 & 1 & 0 & 0 & 1 & 1 & 0 & 1 & 0 & 1 & 1 & 1 & 1 \\
\end{array} \right),
\end{equation}
constructed as in~(\ref{GH}). The set $I_4=\{1,2,3,4,5\}$, or equivalently the set of column vectors $I_4 =\{w_1, w_2, w_3, w_4, w_5\}=\{e_1,e_1+e_2,e_1+e_3,e_1+e_4,e_1+e_5\}$
of $G_{4}$, is an information set for $H_4$.
\end{example}

It is known that the permutation automorphism group $\PAut(H_m)$ of $H_m$ is isomorphic to the general affine group $\AGL(m, 2)$ \cite{MSl}.
Let $\GL(m, 2)$ be the general linear group over $\Z_2$. Recall that $\AGL(m, 2)$  consists of all mappings $\alpha: \mathbb Z_2 ^m \rightarrow \mathbb Z_2^m$ of the form $\alpha(x)=Ax+b$ for $x \in \mathbb Z_2^m$, where $A \in \GL(m, 2)$ and $b \in  \mathbb Z_2^m$, together with the function composition as the group operation. The monomorphism
\begin{equation*}
\begin{array}{cccc}
\varphi: & \AGL(m,2) & \longrightarrow & \GL(m+1, 2)   \\
 &  (b, A) & \longmapsto & \left(\begin{array}{cc} 1 & b \\  \bf{0} & A \\ \end{array} \right) \\
\end{array}
\end{equation*}
defines an isomorphism between $\AGL(m,2)$ and the subgroup of $\GL(m+1, 2)$ consisting of all nonsingular matrices whose first column is $e_1$. Therefore, from now on, we also regard $\PAut(H_m)$ as this subgroup. Note that any matrix $M\in \PAut(H_m)$ can be seen as a permutation of coordinate positions, that is, as an element of $\Sym(2^m)$. By multiplying each column vector $w_i$ of $G_m$ by $M$, we obtain another column vector $w_j=w_iM$, which means that the $i$th coordinate position moves to the $j$th coordinate position, $i,j\in \{1,\ldots,2^m\}$.

Let $M \in \PAut(H_m)$ and let $m_i$ be the $i$th row of $M$, $i\in \{1,\ldots,m+1\}$. We define $M^*$ as the matrix where the first row is $m_1$ and the $i$th row is $m_1+m_i$, $i\in \{2,\ldots,m+1\}$. An $s$-$\PD$-set of size $s+1$ for $H_m$ meets the Gordon-Sch\"onheim bound if $2\leq s \leq f_{m}$. The following theorem provides us a condition on sets of matrices of $\PAut(H_m)$ in order to be $s$-$\PD$-sets of size $s+1$ for $H_m$.

\begin{theorem}
\label{principal}
Let $H_m$ be the binary linear Hadamard code of length $2^m$, with $m\geq 4$. Let $P_s=\{M_i \; : \; 0 \leq i \leq s\}$ be a set of $s+1$ matrices in $\PAut(H_m)$. Then, $P_s$ is an $s$-$\PD$-set of size $s+1$ for $H_m$ with information set $I_m$ if and only if no two matrices $(M_i^{-1})^*$ and $(M_j^{-1})^*$ for $i \neq j$ have a row in common. Moreover, any subset $P_k \subseteq P_s$ of size $k+1$ is a $k$-$\PD$-set for $ k \in \{1, \dots, s\}$.
\end{theorem}

\begin{proof}
Suppose that the set $P_s=\{M_i \; : \;  0 \leq i \leq s\}$ satisfies that no two matrices $(M_i^{-1})^*$ and $(M_j^{-1})^*$ for $i \neq j$ have a row in common. Let $E=\{v_1, \dots, v_s\}\subseteq\{1\}\times \mathbb Z_2 ^{m}$ be a set of $s$ different column vectors of the generator matrix $G_{m}$ regarded as row vectors, which represents a set of $s$ error positions. Assume we cannot move all the error positions to the check set by any element of $P_s$. Then, for each $i \in \{0,\dots,s\}$, there is a $v \in E$ such that $vM_i \in  {I}_{m}$. In other words, there is at least an error position that remains in the information set ${I}_{m}$ after applying any permutation of $P_s$. Note that there are $s+1$ values for $i$, but only $s$ elements in $E$. Therefore, $vM_i\in {I}_{m}$ and $vM_j  \in {I}_{m}$ for some $v \in E$ and $i \neq j$.  Suppose $vM_i=w_r$ and $vM_j=w_t$, for $w_r, w_t \in {I}_{m}$. Then, $v= w_rM_i^{-1}=w_tM_j^{-1}$. Taking into account the form of the vectors in the information set ${I}_{m}=\{w_1, \dots, w_{m+1}\}$, by multiplying for such inverse matrices $M_i^{-1}$ and $M_j^{-1}$, we get the first row or a certain addition between the first row and another row of each matrix. Thus, we obtain that $(M_i^{-1})^*$ and $(M_j^{-1})^*$ have a row in common, contradicting our assumption. Let $P_k \subseteq P_s$ of size $k+1$. If this set satisfies the condition on the inverse matrices and we suppose that it is not a $k$-$\PD$-set, we arrive to a contradiction in the same way as before.

Conversely, suppose that the set $P_s=\{M_i \; : \; 0 \leq i \leq s\}$ forms an $s$-$\PD$-set for $H_m$, but does not satisfy the condition on the inverse matrices. Thus, some $v \in \{w_1,\ldots,w_{2^m}\}$ must be the $r$th row of $(M_i^{-1})^*$ and the  $t$th row of $(M_j^{-1})^*$ for some $r, t \in \{1, \dots, m+1\}$, $i, j \in \{0, \dots, s\}$. In other words, we have that $v=e_r(M_i^{-1})^*=e_t(M_j^{-1})^*$.
Therefore, $v=w_rM_i^{-1}=w_tM_j^{-1}$, where $w_r, w_t \in {I}_{m}$. Finally, we obtain that $vM_i=w_r$ and $vM_j=w_t$. These equalities implies that the vector $v$, which represents an error position, cannot be moved to the check set by the permutations defined by matrices $M_i$ and $M_j$. Let $
L=\{l \; : \; 0 \leq l \leq s, \ l\ne i,j\}$. For each $l\in L$, choose a row $v_l$ of $(M_l^{-1})^*$. It is clear that $v_l=e_t(M_l^{-1})^*=w_tM_l^{-1}$, so $v_lM_l=w_t \in  I_m$. Finally, since some of the $v_l$ may repeat, we obtain a set $E=\{v_l \; : \; l \in L\} \cup \{v\}$ of size at most $s$. Nevertheless, no matrix in $P_s$ will map every member of $E$ into the check set, fact that contradicts our assumption.  \end{proof}
We give now an explicit construction of an $f_m$-${\PD}$-set $\{M_0, \dots, M_{f_m}\}\subseteq \PAut(H_m)$ of minimum size $f_m$+1
for the binary linear Hadamard code $H_{m}$ of length $2^m$. We follow a similar technique to the one described for simplex codes in \cite{FKM}.
\begin{lemma}
\label{primitive}
Let $K=\mathbb Z_2[x]/(f(x))$, where $f(x)$ is a primitive polynomial of degree $m$. If $\alpha$ is a root of $f(x)$, then $\alpha^{i+1}-\alpha^i, \dots, \alpha^{i+m}-\alpha^i$ are linearly independent over $\Z_2$, for all $i \in \{0, \dots, 2^m-2\}$.
\end{lemma}

\begin{proof} 
It is straightforward to see that $\alpha^{i+1}-\alpha^i, \dots, \alpha^{i+m}-\alpha^i$ are linearly independent over $\Z_2$, for all $i \in \{0, \dots, 2^m-2\}$, if and only if $\alpha-1, \dots, \alpha^m-1$ are linearly independent over $\Z_2$, since $\alpha^i \in K\backslash \{0\}$.

Note that $\alpha^m-1=\sum_{j=1}^{m-1}\mu_j \alpha^j$, where the sum has and odd number of nonzero terms, since $f(x)$ is irreducible. 
Let $\mathbf{\mu}=(\mu_1, \dots, \mu_{m-1})\in \Z_2^{m-1}$. Note that in vectorial notation $\alpha^j-1=e_1+e_ j$, $j \in \{1, \dots, m-1\}$ and $\alpha^m-1=\sum_{j=1}^{m-1} \mu_j e_{j+1}$. Finally, it is easy to see that the $m \times m$ binary  matrix
$$\left(\begin{array}{cc}
\mathbf{1}      & \operatorname{Id}_{m-1}  \\
0               & \mu                      \\
\end{array}\right),$$
which has as rows $\alpha-1, \dots, \alpha^m-1$, has determinant $\sum_{j=1}^{m-1}\mu_j =1 \neq 0$. \end{proof}
For $i \in \{1, \dots, f_m\}$, consider the following $(m+1) \times (m+1)$ binary matrices:
\begin{equation*}
N_0= \left(\begin{array}{c  c}
1         & 0 \\
0         & 1  \\
\vdots    & \vdots \\
0         & \alpha^{m-1} \\
\end{array} \right) \quad {\rm and}
\quad
N_i= \left(\begin{array}{c  c}
1      & \alpha^{(m+1)i-1} \\
0      & \alpha^{(m+1)i}- \alpha^{(m+1)i-1}  \\
\vdots & \vdots \\
0      & \alpha^{(m+1)i+m-1}- \alpha^{(m+1)i-1}   \\
\end{array} \right).
\end{equation*}

\begin{theorem}
Let $P_{f_m}=\{M_i \; : \; 0 \leq i \leq f_m\}$, where $M_ i=N_i^{-1}$. Then, $P_{f_m}$ is an $f_m$-${\PD}$-set of size $f_m+1$ for the binary linear Hadamard code $H_m$ of length $2^m$ with information set $I_m$.
\end{theorem}

\begin{proof}
Clearly, $N_0 \in  \PAut(H_m)$, since it is the identity matrix.
By Lemma~\ref{primitive}, $N_i \in \PAut(H_m)$ for all $i \in \{1, \dots, f_m\}$. Moreover, rows of matrices $N_0^*, \dots, N_{f_m}^*$ form the set $\{(1, a) \; : \; a \in \{0, 1, \alpha, \dots, \alpha^{f_m(m+1)+m-1}\}\}$. The elements of such set are different since $\alpha$ is primitive and $f_m(m+1)+m-1 \leq 2^m-2 $. Theorem~\ref{principal} completes the proof. \end{proof}

\begin{example}
\label{ejemplo}
Let $H_4$ be the binary linear Hadamard code of length 16 with generator matrix (\ref{GH4}). Let 
$K=\mathbb Z_2[x]/(x^4+x+1)$  and $\alpha$ a root of $x^4+x+1$. 
Matrices $N_0= \operatorname{Id_5}$,
\begin{equation*}
N_1=
\left(\begin{array}{cc}
1      & \alpha^{4}               \\
0      & \alpha^{5}-\alpha^{4}    \\
0      & \alpha^{6}-\alpha^{4}    \\
0      & \alpha^{7}-\alpha^{4}    \\
0      & \alpha^{8}-\alpha^{4}    \\
\end{array} \right)
=
\left(\begin{array}{cccccc}
1  & 1  & 1 & 0 & 0  \\
0  & 1 & 0 & 1  & 0   \\
0  & 1 & 1 & 1  & 1   \\
0  & 0 & 0& 0 & 1   \\
0  & 0 & 1 & 1 & 0 \\
\end{array} \right), \quad 
N_2=
\left(\begin{array}{cc}
1      & \alpha^{9}               \\
0      & \alpha^{10}-\alpha^{9}    \\
0      & \alpha^{11}-\alpha^{9}    \\
0      & \alpha^{12}-\alpha^{9}    \\
0      & \alpha^{13}-\alpha^{9}    \\
\end{array} \right)
=
\left(\begin{array}{cccccc}
1 & 0 & 1 & 0 & 1  \\
0 & 1 & 0 & 1 & 1   \\
0 & 0 & 0 & 1 & 0   \\
0 & 1 & 0 & 1 & 0   \\
0 & 1 & 1 & 1 & 0 \\
\end{array} \right),
\end{equation*}
where $\operatorname{Id}_5$ is the $5\times 5$ identity matrix, are elements of $\PAut(H_4)$ 
and $P_2=\{N_0^{-1}, N_1^{-1}, N_2^{-1}\}$
is a 2-${\PD}$-set of size 3 for $H_4$. It is straightforward to chech that matrices $N_0^*$,
\begin{equation*}
N_1^*=
\left(\begin{array}{cccccc}
1  & 1  & 1 & 0 & 0  \\
1  & 0 & 1 & 1  & 0   \\
1  & 0 & 0 & 1  & 1   \\
1  & 1 & 1 & 0 & 1   \\
1  & 1 & 0 & 1 & 0 \\
\end{array} \right), \quad {\rm and}
\quad
N_2^*=
\left(\begin{array}{cccccc}
1 & 0 & 1 & 0 & 1  \\
1 & 1 & 1 & 1 & 0   \\
1 & 0 & 1 & 1 & 1   \\
1 & 1 & 1 & 1 & 1   \\
1 & 1 & 0 & 1 & 1 \\
\end{array} \right),
\end{equation*}
have no rows in common. Finally, no $s$-$\PD$-set of size $s+1$ can be found for $s \geq 3$ since $f_{4}=2$.
\end{example}

Let $S$ be an $s$-$\PD$-set of size $s+1$. The set $S$ is a {\em nested} $s$-$\PD$-set if there is an ordering of the elements of $S$, $S=\{\sigma_0, \dots, \sigma_s\}$, such that $S_i=\{\sigma_0,\ldots, \sigma_i\} \subseteq S$ is an $i$-$\PD$-set of size $i+1$, for all $i \in \{0, \dots, s\}$. Note that $S_i \subset S_j$ if $0 \leq i < j \leq s$ and $S_s=S$. From Theorem~\ref{principal}, we have two important consequences. The first one is related to how to obtain nested $s$-$\PD$-sets and the second one provides another proof of Lemma~\ref{bound}.

\begin{corollary}
Let $m$ be an integer, $m\geq 4$. If $P_s$ is an $s$-$\PD$-set of size $s+1$ for the binary linear Hadamard code $H_m$, then any ordering of the elements of $P_s$ gives nested $k$-$\PD$-sets for $k \in \{1, \dots, s\}$.
\end{corollary}

\begin{corollary}
Let $m$ be an integer, $m \geq 4$. If $P_s$ is an $s$-$\PD$-set of size $s+1$ for the binary linear Hadamard code $H_m$, then  $s \leq \left \lfloor{\frac{2^m-m-1}{1+m}} \right \rfloor$.
\end{corollary}

\begin{proof}
Following the condition on sets of matrices to be $s$-$\PD$-sets of size $s+1$, given by Theorem~\ref{principal}, we have to obtain certain $s+1$ matrices with no rows in common. Note that the number of possible vectors of length $m+1$ over $\mathbb Z_2$ with 1 in the first coordinate is $2^m$. Thus, taking this fact into account and counting the number of rows of each one of these $s+1$ matrices, we have that $(s+1)(m+1) \leq 2^m$, so $s+1 \leq \frac{2^m}{m+1}$ and finally $s \leq \left \lfloor{\frac{2^m-m-1}{1+m}} \right \rfloor$. 
\end{proof}

\section{Recursive construction of $s$-$\PD$-sets for binary linear Hadamard codes}
\label{sec:binarylinearRecursive}

In this section, given an $s$-$\PD$-set of size $l$ for the binary linear Hadamard code $H_m$ of length $2^m$, where $l \geq s+1$, we show how to construct recursively an $s$-$\PD$-set of the same size for $H_{m'}$ of length $2^{m'}$ for all $m'>m$.

Given a matrix $M \in \PAut(H_m)$ and an integer $\kappa \geq 1$, we define the matrix $M(\kappa) \in \PAut(H_{m+\kappa})$ as
\begin{equation}
\label{M(v)}
M(\kappa)=\left(\begin{array}{cc}
M & \mathbf{0} \\
\mathbf{0} & \operatorname{Id}_{\kappa} \\
\end{array}\right),
\end{equation}
where $\operatorname{Id}_{\kappa}$ denotes the $\kappa \times \kappa$ identity matrix.

\begin{proposition}
\label{recursivo}
Let $m$ be an integer, $m \geq 4$, and $P_s=\{M_i \; : \; 0 \leq i \leq s\}$ be an $s$-$\PD$-set of size $s+1$ for $H_m$ with information set $I_m$. Then, $Q_s=\{(M^{-1}_i(\kappa))^{-1} \; : \; 0 \leq i \leq s\}$ is an $s$-$\PD$-set of size $s+1$ for $H_{m+\kappa}$ with information set $I_{m+\kappa}$, for any $\kappa \geq 1$.
\end{proposition}

\begin{proof}
Since $P_s$ is an $s$-$\PD$-set for $H_m$, matrices $(M^{-1}_1)^*, \dots, (M^{-1}_s)^*$ have no rows in common by Theorem~\ref{principal}. Therefore, it is straightforward to check that matrices $(M^{-1}_1(\kappa))^*, \dots,(M^{-1}_s(\kappa))^*$  have no rows in common either. Moreover, $M_i^{-1}(\kappa) \in \operatorname{PAut}(H_{m+\kappa}),$ for all $i \in \{1, \dots, s\}$. Thus, applying again Theorem~\ref{principal}, we have that $Q_s$ is an $s$-$\PD$-set for $H_{m+\kappa}$. 
\end{proof}

It is important to note that the bound $f_{{m+1}}$ for $H_{m+1}$ cannot be achieved recursively from an $s$-$\PD$-set for $H_{m}$, since the above recursive construction works for a given fixed $s$, increasing the length of the Hadamard code.


The above recursive construction only holds when the size of the $s$-$\PD$-set is exactly $s+1$. Now, we will show a second recursive construction which holds when the size of the $s$-$\PD$-set is any integer $l$, $l\geq s+1$.
In this case, the elements of  $\operatorname{PAut}(H_{m})$ will be regarded as permutations of coordinate positions, that is,
as elements of $ \operatorname{Sym}(2^m)$ instead of matrices of $\GL(m+1,2)$.

It is well known that a generator matrix $G_{m+1}$ for the binary linear Hadamard code $H_{m+1}$ of length $2^{m+1}$ can be constructed as follows:
\begin{equation}
\label{doubleZ2}
G_{m+1} =
\left(\begin{array}{c c}
 G_{m} &   G_{m} \\
\mathbf{0}  & \mathbf{1} \\
\end{array} \right),
\end{equation}
where $G_m$ is a generator matrix for the binary linear Hadamard code $H_{m}$ of length $2^{m}$. Given two permutations $\sigma_1 \in \operatorname{Sym}(n_1)$ and  $\sigma_2 \in \operatorname{Sym}(n_2)$, we define $(\sigma_1 | \sigma_2) \in \operatorname{Sym}(n_1+n_2)$, where  $\sigma_1$ acts on the coordinates $\{1, \dots, n_1\}$ and $\sigma_2$ on $\{n_1+1, \dots, n_1+n_2\}$.

\begin{proposition}
\label{doublepdsetZ2}
Let $m$ be an integer, $m \geq 4$, and $S$ be an $s$-$\PD$-set of size $l$ for $H_{m}$ with information set $I$. Then, $( S |  S)=\{(\sigma | \sigma) : \sigma \in S\}$ is an $s$-$\PD$-set of size $l$ for $H_{m+1}$ constructed from (\ref{doubleZ2}),
with any information set $I'=I \cup \{i+ 2^{m}\}, \; i\in I$.
\end{proposition}

\begin{proof}
Since $I$ is an information set for $H_m$, we have that $|(H_m)_I|=2^{m+1}$. Since $H_{m+1}$ is constructed from (\ref{doubleZ2}), it follows that $H_{m+1}=\{(x, x), (x, \bar x) : x \in H_{m}\}$, where $\bar x$ is the complementary vector of $x$. A vector and its complementary have different values in each coordinate, so $|(H_{m+1})_{I \cup \{i\}})|=2^{m+2}$, for all $i \in \{2^m+1, \dots, 2^{m+1}\}$. Thus, any set of the form $I'=I \cup \{i+ 2^{m}\},$ $i\in I$, is an information set for $H_{m+1}$.

If $\sigma \in \operatorname{PAut}(H_{m})$, then $\sigma(x)=z \in H_{m}$ for all $x \in H_{m}$. Therefore, since $(\sigma | \sigma)(x, x)=(z, z)$ and $(\sigma | \sigma)(x, \bar{x})=(z, z+ \sigma(\mathbf{1}))=(z, \bar{z})$, we can conclude that $(\sigma | \sigma) \in \operatorname{PAut}(H_{m+1})$.

Let $e=(a,b) \in \Z_2^{2n}$, where $a=(a_1, \dots, a_n)$, $b=(b_1, \dots, b_n) \in \mathbb Z_2^{n}$, and $n=2^m$. Finally, we will prove that for every $e \in \mathbb Z_2^{2n}$ with $\operatorname{wt}(e) \leq s$, there is $(\sigma | \sigma) \in (S | S)$ such that $(\sigma | \sigma)(e)_{I'} = \mathbf{0}$.
Let $c=(c_1,\dots,c_n)$ be the binary vector defined as follows: $c_i=1$ if and only if $a_i=1$ or $b_i=1$, for all $i\in \{1,\dots,n\}$.
Note that $\wt(c) \leq s$, since $\wt(e)\leq s$. Taking into account that $S$ is an $s$-$\PD$-set with respect to $I$, there is $\sigma \in S$ such that $\sigma(c)_{I}=\mathbf{0}$. Therefore, we also have that $(\sigma |\sigma)(a,b)_{I \cup J}=\mathbf {0}$, where $J=\{i+n: i \in I\}$. The result follows trivially since $I'\subseteq I \cup J$. 
\end{proof}

\section{Finding $s$-$\PD$-sets of size $s+1$ for Hadamard $\Z_4$-linear codes}
\label{sec:Z4linear}

For any $m\geq 3$ and each $\delta\in \{1, \dots, \left \lfloor \frac{m+1}{2} \right \rfloor\}$, there is an unique (up to equivalence) $\Z_4$-linear Hadamard code of length $2^{m}$ which is the Gray map image of a quaternary linear code of length $\beta=2^{m-1}$ and type $2^\gamma 4^\delta$, where $m=\gamma+2\delta-1$. Moreover, for a fixed $m$, all these codes are pairwise nonequivalent, except for $\delta=1$ and $\delta=2$, since these ones are equivalent to the binary linear Hadamard code of length $2^m$ \cite{K}. Therefore, the number of nonequivalent $\Z_4$-linear Hadamard codes of length $2^m$ is $\left \lfloor{\frac{m-1}{2}} \right \rfloor $ for all $m \geq 3$.
Note that when $\delta \geq 3$, the $\Z_4$-linear Hadamard codes are nonlinear.

Let $\mathcal H_{\gamma, \delta}$ be the quaternary linear Hadamard code of length $\beta=2^{m-1}$ and type $2^{\gamma}4^{\delta}$, where $m=\gamma+2\delta-1$, and let $H_{\gamma, \delta}=\Phi(\mathcal H_{\gamma, \delta})$ be the corresponding $\mathbb Z_4$-linear code of length $2\beta=2^m$.
A generator matrix $\mathcal G_{\gamma, \delta}$ for the code $\mathcal H_{\gamma, \delta}$ can be constructed by using the following recursive constructions:
\begin{align}
\label{double}
\mathcal G_{\gamma+1, \delta} & =
\left(\begin{array}{c c}
\mathcal G_{\gamma, \delta} &  \mathcal G_{\gamma, \delta} \\
\mathbf{0}  & \mathbf{2} \\
\end{array} \right),
\\
\label{quadruple}
\mathcal G_{\gamma, \delta+1} & = \left(\begin{array}{c c c c}
\mathcal G_{\gamma, \delta} &  \mathcal G_{\gamma, \delta}  &  \mathcal G_{\gamma, \delta}  &  \mathcal G_{\gamma, \delta}\\
\mathbf{0}  & \mathbf{1} & \mathbf{2}  & \mathbf{3}  \\
\end{array} \right),
\end{align}
starting from $\mathcal G_{0, 1}=(1)$. We first obtain $\mathcal G_{0, \delta}$ from $\mathcal G_{0,1}$ by using recursively $\delta$ times construction (\ref{quadruple}). Then, $\mathcal G_{\gamma, \delta}$ is managed from $\mathcal G_{0, \delta}$ by using $\gamma$ times construction (\ref{double}).
Note that the rows of order four remain in the upper part of $\mathcal G_{\gamma, \delta}$ while those of order two stay in the lower part.

A set $\mathcal I=\{i_1, \dots, i_{\gamma+\delta}\} \subseteq \{1,\ldots,\beta \}$ of $\gamma+\delta$ coordinate positions is said to be a {\em quaternary information set} for a quaternary linear code $\mathcal C$ of type $2^\gamma 4^\delta$ if $|\mathcal C_{\mathcal I}|=2^{\gamma}4^{\delta}$.  If the coordinates in $\mathcal I$ are ordered in such a way that $|\mathcal C_{\{i_1,\dots, i_{\delta}\}}|=4^{\delta}$, it is easy to see that the set $\Phi(\mathcal I)$, defined as
$$\Phi(\mathcal I)=\{2i_1-1, 2i_1, \dots, 2i_{\delta}-1, 2i_{\delta}, 2i_{\delta+1}-1, \dots, 2i_{\delta+\gamma}-1\},$$
is an information set for $C=\Phi(\mathcal C)$.
For example, the set $\mathcal I=\{1\}$ is a quaternary information set for $\mathcal H_{0,1}$, so $\Phi(\mathcal I)=\{1,2\}$ is an information set for $H_{0,1}=\Phi(\mathcal H_{0,1})$. In general, there is not an unique way to obtain a quaternary information set for the code
$\mathcal H_{\gamma,\delta}$. The following result provides a recursive and simple form to obtain such a set.

\begin{proposition}
\label{infoset}
Let  $\mathcal I$ be a quaternary information set for the quaternary linear Hadamard code $\mathcal H_{\gamma, \delta}$ of length $\beta=2^{m-1}$ and type $2^{\gamma}4^{\delta}$, where $m=\gamma+2\delta-1$. Then $\mathcal I \cup \{\beta+1\}$ is a quaternary information set for the codes $\mathcal H_{\gamma+1, \delta}$ and $\mathcal H_{\gamma, \delta+1}$, which are obtained from $\mathcal H_{\gamma, \delta}$ by applying (\ref{double}) and (\ref{quadruple}), respectively.
\end{proposition}	

\begin{proof}
Since  $|\mathcal H_{\gamma+1, \delta}|=2^{\gamma+1}4^{\delta}$ and  $|\mathcal H_{\gamma, \delta+1}|=2^{\gamma}4^{\delta+1}$, it is clear that a quaternary information set for codes $\mathcal H_{\gamma+1, \delta}$ and $\mathcal H_{\gamma, \delta+1}$ should have $\gamma+\delta+1=|\mathcal I|+1$  coordinate positions.

Taking into account that $\mathcal H_{\gamma, \delta+1}$ is constructed from (\ref{quadruple}), we have that $\mathcal H_{\gamma , \delta+1}=\{(u, u, u, u), (u, u+\mathbf{1}, u+\mathbf{2}, u+\mathbf{3}), (u, u+ \mathbf{2}, u, u+\mathbf{2}), (u, u+\mathbf{3}, u+\mathbf{2}, u+\mathbf{1}) : u \in \mathcal H_{\gamma, \delta}\}$.  Vectors $u$, $u+\mathbf{1}, u+\mathbf{2}$, and $u+\mathbf{3}$  have different values in each coordinate, so $|(\mathcal H_{\gamma, \delta+1})_{\mathcal I \cup \{i\}}|=2^{\gamma}4^{\delta+1}$ for all $i \in \{\beta+1, \dots, 2\beta, 3\beta+1, \dots, 4\beta \}$. In particular, $\mathcal I \cup \{\beta+1\}$ is a quaternary information set for $\mathcal H_{\gamma, \delta+1}$.

A similar argument holds for $\mathcal H_{\gamma +1, \delta}$. Since $\mathcal H_{\gamma +1, \delta}$ is constructed from (\ref{double}), we have that $\mathcal H_{\gamma +1, \delta}=\{(u, u), (u, u+ \mathbf{2}) : u \in \mathcal H_{\gamma, \delta}\}$.  Vectors $u$ and $u+\mathbf{2}$  have different values in each coordinate, so $|(\mathcal H_{\gamma+1, \delta})_{\mathcal I \cup \{i\}}|=2^{\gamma+1}4^\delta$ for all $i \in \{\beta+1, \dots, 2\beta \}$. Therefore, $\mathcal I \cup \{\beta+1\}$ is a quaternary information set for $\mathcal H_{\gamma+1, \delta}$. 
\end{proof}

Despite the fact that the quaternary information set $\mathcal I \cup \{\beta+1\}$, given by Proposition~\ref{infoset}, is the same for $\mathcal H_{\gamma+1, \delta}$ and $\mathcal H_{\gamma, \delta+1}$, the information set for the corresponding binary codes $H_{\gamma+1, \delta}$ and $H_{\gamma, \delta+1}$ are $I'=\Phi(\mathcal I) \cup \{2\beta+1\}$ and $I''=\Phi(\mathcal I) \cup \{2\beta+1, 2\beta+2\}$, respectively. As for binary linear codes, we can label the $i$th coordinate position of a quaternary linear code $\mathcal C$, with the $i$th column of a generator matrix $\mathcal G$ of $\mathcal C$. Thus, any quaternary information  set $\mathcal I$ for $\mathcal C$ can also be considered as a set of vectors representing the positions in $\mathcal I$. 	Then, by Proposition \ref{infoset}, we have that the set $\mathcal I_{\gamma, \delta}=\{e_1, e_1+e_2, \dots, e_1+e_{\delta},  e_1+2e_{\delta+1}, \dots, e_1+2e_{\gamma+\delta}\}$
is a suitable quaternary information set for the code $\mathcal H_{\gamma,\delta}$. Depending on the context, $\mathcal I_{\gamma, \delta}$ will be considered as a subset of $\{1, \dots, \beta\}$ or as a subset of $\{1\} \times \mathbb Z_4 ^{\delta-1} \times \{0, 2\}^{\gamma}$.

\begin{example}
The quaternary linear Hadamard code $\mathcal H_{0,3}$ of length 16 can be generated by the matrix
$$
\mathcal G_{0,3}=\left(\begin{array}{c c c c c c c c c c c c c c c c}
1  & 1 & 1 & 1    &  1  & 1 & 1 & 1          & 1  & 1 & 1 &  1      & 1  & 1 & 1 &  1\\
0  & 1 & 2 & 3    &  0  & 1 & 2 & 3         & 0  & 1 & 2 & 3         & 0  & 1 & 2 & 3  \\
0  & 0 & 0 & 0    & 1 & 1 & 1 & 1           & 2 & 2 & 2 & 2          & 3 & 3 & 3 & 3\\
\end{array} \right),
$$
obtained by applying two times construction (\ref{quadruple}) over $\mathcal G_{0,1}=(1)$. The set $\mathcal I_{0,3}=\{1, 2, 5\}$, or equivalently the set of column vectors $\mathcal I_{0,3}=\{(1, 0, 0), (1, 1, 0), (1, 0, 1)\}$ of $\mathcal G_{0,3}$, is a quaternary information set for $\mathcal H_{0,3}$.  By applying constructions (\ref{double}) and (\ref{quadruple}) over $\mathcal G_{0,3}$, we obtain that matrices
\begin{align*}
\mathcal G_{1, 3} & =
\left(\begin{array}{c c}
\mathcal G_{0,3} &  \mathcal G_{0,3} \\
\mathbf{0}  & \mathbf{2} \\
\end{array} \right),
\\
\mathcal G_{0,4} & = \left(\begin{array}{c c c c}
\mathcal G_{0,3} &  \mathcal G_{0,3}  &  \mathcal G_{0,3}  &  \mathcal G_{0,3}\\
\mathbf{0}  & \mathbf{1} & \mathbf{2}  & \mathbf{3}  \\
\end{array} \right),
\end{align*}
generate the quaternary linear Hadamard codes $\mathcal H_{1,3}$ and $\mathcal H_{0,4}$ of length 32 and 64, respectively. By Propositions \ref{infoset}, it follows that $\mathcal I_{0,3} \cup \{17\}=\{1,2,5,17\}$  is a quaternary information set for $\mathcal H_{1,3}$ and $\mathcal H_{0,4}$. Despite the fact that the quaternary information set is the same for both codes $\mathcal H_{1,3}$ and $\mathcal H_{0,4}$, it is important to note that in terms of column vectors representing these positions, we have that $\mathcal I_{1,3}=\{(1, 0, 0, 0), (1, 1, 0, 0), (1, 0, 1, 0), (1, 0, 0, 2)\}$ and $\mathcal I_{0,4}=\{(1, 0, 0, 0), (1, 1, 0, 0),  (1, 0, 1, 0), (1, 0, 0, 1)\}$. Finally,   $I'=\Phi(\mathcal I_{0,3}) \cup \{33\}=$ $\{1, 2, 3, 4, 9, 10, 33\}$ and  $I''=\Phi(\mathcal I_{0,3}) \cup \{33, 34\}=\{1, 2, 3, 4, 9, 10, 33,34\}$ are information sets for the $\Z_4$-linear Hadamard codes $H_{1,3}$ and $H_{0,4}$, respectively.
\end{example}

Let $\mathcal C$ be a quaternary linear code of length $\beta$ and type $2^{\gamma}4^{ \delta}$, and let $C=\Phi(\mathcal C)$ be the corresponding $\mathbb Z_4$-linear code of length $2\beta$.  Let $\Phi: \operatorname{Sym}(\beta) \rightarrow \operatorname{Sym}(2\beta)$ be the map defined as
$$\Phi(\tau)(i)=\left\{\begin{array}{l  l}
2\tau(i/2), & \mathrm{if} \; i \; \mathrm{is \; even}, \\
2\tau((i+1)/2) -1 & \mathrm{if} \; i \; \mathrm{is \; odd}, \\
\end{array}\right.$$
for all $\tau \in \operatorname{Sym}(\beta)$ and  $i \in \{1, \dots, 2\beta \}$. Given a subset $\mathcal S \subseteq \operatorname{Sym}(\beta)$, we define the set $\Phi(\mathcal S)=\{\Phi(\tau) : \tau \in \mathcal S\}\subseteq \operatorname{Sym}(2\beta)$. It is easy to see that if $\mathcal S \subseteq \PAut(\mathcal C) \subseteq \Sym(\beta)$, then $\Phi(\mathcal S) \subseteq \operatorname{PAut}(C) \subseteq \operatorname{Sym}(2\beta)$.

Let $\GL(k, \Z_4)$ denote the general linear group of degree $k$ over $\Z_4$ and
let $\mathcal L$ be the set consisting of all matrices over $\Z_4$ of the following form:
$$\left(\begin{array}{ccc}
1 & \eta & 2\theta \\
\mathbf{0} & A & 2X \\
\mathbf{0} & Y & B \\
\end{array}\right), $$
where $A \in \GL(\delta-1, \mathbb Z_4), B \in \GL(\gamma, \mathbb Z_4),$ $X$ is a matrix over $\mathbb Z_4$ of size $(\delta-1) \times \gamma$, $Y$ is a matrix over $\mathbb Z_4$ of size $\gamma\times (\delta-1)$, $\eta \in \mathbb Z_4^{\delta-1}$ and $\theta\in \mathbb Z_4^{\gamma}$.

\begin{lemma}
\label{Lsubgroup}
The set $\mathcal L$ is a subgroup of $\operatorname{GL}(\gamma+ \delta, \mathbb Z_4)$.
\end{lemma}
\begin{proof}
We first need to check that $\mathcal{L}\subseteq \operatorname{GL}(\gamma+ \delta, \mathbb Z_4)$, in other words, that  $\det(\mathcal M)\in \{1, 3\}$ (that is, an unit of $\mathbb Z_4$) for all $\mathcal M \in \mathcal L$. Note that if $\mathcal M' \in \operatorname{GL}(k, \mathbb Z_4)$, then $\mathcal M =\mathcal M' +2 \mathcal R \in \operatorname{GL}(k, \mathbb Z_4)$. Thus, since $\operatorname{det}(\mathcal M') \in \{1,3\}$, we have that $\operatorname{det}(\mathcal M) \in \{1,3\}$, where
$$\mathcal M'=\left(\begin{array}{ccc}
1 & \eta & \mathbf{0} \\
\mathbf{0} & A & \mathbf{0} \\
\mathbf{0} & Y & B \\
\end{array}\right). $$
It is straightforward to check that $\mathcal M \mathcal N \in \mathcal L$ for all $\mathcal M, \mathcal N \in \mathcal L$. 
\end{proof}

Let $\zeta$ be the map from $\mathbb Z_4$ to $\mathbb Z_4$ which is the usual modulo two map composed with inclusion from $\mathbb Z_2$ to $\mathbb Z_4$, that is $\zeta(0)=\zeta(2)=0, \zeta(1)=\zeta(3)=1$. This map can be extended to matrices over $\mathbb Z_4$ by applying $\zeta$ to each one of their entries. Let $\pi$ be the map from $\mathcal L$ to $\mathcal L$ defined as
$$\pi(\mathcal M)=\left(\begin{array}{ccc}
1 & \eta & 2\theta \\
\mathbf{0} & A & 2X \\
\mathbf{0} & \zeta(Y) & \zeta(B) \\
\end{array}\right), $$
and let $\pi(\mathcal L)=\{\pi(\mathcal M): \mathcal M \in \mathcal L\} \subseteq \operatorname{GL}(\gamma+\delta, \mathbb Z_4)$. By Lemma~\ref{Lsubgroup}, it is clear that $\pi(\mathcal L)$ is a group with the operation $*$ defined as $\mathcal M * \mathcal N=\pi(\mathcal M \mathcal N)$ for all $\mathcal M, \mathcal N \in \pi(\mathcal L)$. By the proof of Theorem 2 in \cite{KVi}, it is easy to see that the permutation automorphism group $\PAut(\mathcal H_{\gamma,\delta})$ of $\mathcal H_{\gamma, \delta}$  is isomorphic to $\pi(\mathcal L)$. Thus, from now on, we identify $\operatorname{PAut}(\mathcal H_{\gamma, \delta})$ with this group. Recall that we can label the $i$th coordinate position of $\mathcal H_{\gamma, \delta}$ with the $i$th column vector $w_i$ of the generator matrix $\mathcal G_{\gamma, \delta}$ constructed via (\ref{double}) and (\ref{quadruple}), $i\in \{1,\ldots,\beta \}$. Therefore, again, any matrix $\mathcal M \in \PAut(\mathcal H_{\gamma, \delta})$ can be seen as a permutation of coordinate positions $\tau \in \operatorname{Sym}(\beta)$, such that $\tau(i)=j$ as long as $w_j=w_i\mathcal M$, $i,j \in \{1, \dots, \beta\}$. For any $\mathcal M \in \operatorname{PAut}(\mathcal H_{\gamma, \delta})$, we define $\Phi(\mathcal M)=\Phi(\tau)\in \operatorname{Sym}(2\beta)$, and for any $\mathcal P \subseteq \operatorname{PAut}(\mathcal H_{\gamma, \delta})$, we consider $\Phi(\mathcal P)=\{\Phi (\mathcal M): \mathcal M \in \mathcal P\} \subseteq \Sym(2\beta)$.

\begin{lemma}
\label{Auxiliar}
Let $\mathcal H_{\gamma, \delta}$ be the quaternary linear Hadamard code of length $\beta$ and type $2^{\gamma}4^{\delta}$ and let $\mathcal P \subseteq \operatorname{PAut}(\mathcal H_{\gamma, \delta})$. Then,  $\Phi(\mathcal P)$ is an $s$-$\PD$-set for $H_{\gamma, \delta}$ with information set $\Phi(\mathcal I_{\gamma, \delta})$ if and only if for every $s$-set $\mathcal E$ of column vectors of $\mathcal G_{\gamma, \delta}$
there is $\mathcal M \in \mathcal P$ such that $\{ g\mathcal M : g \in \mathcal E \} \cap \mathcal I_{\gamma, \delta}=\emptyset$.
\end{lemma}

\begin{proof}
If $\Phi(\mathcal P)$ is an $s$-$\PD$-set with respect to the information set $\Phi(\mathcal I_{\gamma, \delta})$,
then for every $s$-set $E\subseteq \{1,\ldots,2\beta\}$, there is $\tau \in \mathcal P \subseteq \Sym(\beta)$ such that $\Phi(\tau)(E) \cap \Phi(\mathcal I_{\gamma, \delta}) = \emptyset$. For every $s$-set $\mathcal E \subseteq \{1,\ldots, \beta\}$, let $E_o=\{2i-1 : i \in \mathcal E\}$. We know that there is $\tau \in \mathcal P$ such that $\Phi(\tau)(E_o) \cap \Phi(\mathcal I_{\gamma, \delta}) = \emptyset$.
By the definition of $\Phi$,
we also have that $\tau(\mathcal E) \cap \mathcal I_{\gamma, \delta} = \emptyset$, which is equivalent to the statement.

Conversely, we assume that for every $s$-set $\mathcal E \subseteq \{1, \dots, \beta\}$, there is $\tau \in \mathcal P \subseteq \Sym(\beta)$ such that $\tau(\mathcal E) \cap \mathcal I_{\gamma, \delta}=\emptyset$. For every $s$-set $E\subseteq \{1,\ldots,2\beta\}$, let $\mathcal E_o$ be an $s$-set such that $\{ i : \varphi_1(i)\in E \textrm{ or } \varphi_2(i) \in E \} \subseteq \mathcal E_o$, where $\varphi_1(i)=2i-1$ and $ \varphi_2(i)=2i$. Since there is $\tau \in \mathcal P$ such that $\tau(\mathcal E_o) \cap \mathcal I_{\gamma, \delta}=\emptyset$, we have that $\Phi(\tau)(E) \cap \Phi(\mathcal I_{\gamma, \delta}) = \emptyset$. 
\end{proof}

Let $\mathcal M \in \PAut(\mathcal H_{\gamma, \delta})$ and let $m_i$ be the $i$th row of $\mathcal M$, $i\in \{1,\ldots,\delta+\gamma\}$. We define $\mathcal M^*$ as the matrix where the first row is $m_1$ and the $i$th row is $m_1+m_i$
for $i\in \{2,\ldots, \delta\}$ and $m_1+2m_i$ for $i\in \{\delta+1,\ldots, \delta+\gamma\}$.

\begin{theorem}
\label{principalZ4}
Let $\mathcal H_{\gamma, \delta}$ be the quaternary linear Hadamard code of type $2^{\gamma}4^{\delta}$.  Let $\mathcal P_s=\{\mathcal M_i \; : \; 0 \leq i \leq s\}$ be a set of $s+1$ matrices in $\PAut(\mathcal H_{\gamma, \delta})$. Then, $\Phi(\mathcal P_s)$ is an $s$-$\PD$-set of size $s+1$ for $H_{\gamma, \delta}$ with information set $\Phi(\mathcal I_{\gamma, \delta})$ if and only if no two matrices $(\mathcal M_i^{-1})^*$ and $(\mathcal M_j^{-1})^*$ for $i \neq j$ have a row in common.
\end{theorem}

\begin{proof}
By Lemma~\ref{Auxiliar} and following a similar argument as in the proof of Theorem~\ref{principal}. 
\end{proof}


\begin{corollary}
Let $\mathcal P_s$ be a set of $s+1$ matrices in $\PAut(\mathcal H_{\gamma, \delta})$.
If $\Phi(\mathcal P_s)$ is an $s$-$\PD$-set of size $s+1$ for $H_{\gamma, \delta}$, then any ordering of elements in $\Phi(\mathcal P_s)$ provides nested $k$-$\PD$-sets for $k \in \{1, \dots, s\}$.
\end{corollary}

\begin{corollary}
Let $\mathcal P_s$ be a set of $s+1$ matrices in $\PAut(\mathcal H_{\gamma, \delta})$.
If $\Phi(\mathcal P_s)$ is an $s$-$\PD$-set of size $s+1$ for  $H_{\gamma, \delta}$,
then $s \leq f_{\gamma, \delta}$, where $$f_{\gamma, \delta}=\left \lfloor{\frac{2^{\gamma+2\delta-2}-\gamma-\delta}{\gamma+\delta}} \right \rfloor.$$
\end{corollary}

\begin{proof}
Following the condition on sets of matrices to be $s$-$\PD$-sets of size $s+1$, given by Theorem~\ref{principalZ4}, we have to obtain certain $s+1$ matrices with no rows in common. Since the rows of length $\delta+\gamma$ must have 1 in the first coordinate, and elements from $\{0,2\}$ in the last $\gamma$ coordinates, the number of possible rows is $4^{\delta-1}2^\gamma=2^{\gamma+2\delta-2}$. Thus, taking this fact into account and counting the number of rows of each one of these $s+1$ matrices, we have that $(s+1)(\gamma+\delta) \leq 2^{\gamma+2\delta-2}$, and the result follows. 
\end{proof}


We give now an explicit construction of an $f_{0, \delta}$-${\PD}$-set of size $f_{0, \delta}+1$ for $H_{0, \delta}$. Let $\mathcal R=\operatorname{GR}(4^{\delta-1})$ be the Galois extension of dimension $\delta-1$ over $\Z_4$. It is known that $\mathcal R$ is isomorphic to $\Z_4[x]/(h(x))$, where $h(x)$ is a monic basic irreducible polynomial of degree $\delta-1$. 
Let $f(x)\in \mathbb Z_2[x]$ be a primitive polynomial of degree $\delta-1$. Let $\ell=2^{\delta-1}-1$. There is a unique primitive basic irreducible polynomial $h(x)$ dividing $x^{\ell}-1$ in $\mathbb Z_4[x]$. Let $\T=\{0, 1, \alpha, \dots, \alpha^{\ell-1}\}\subseteq \mathcal R$, where $\alpha $ is a root of $h(x)$.
It is well known that any $r\in \R$ can be written uniquely as $r=a+2b,$ where $a, b \in \T$. We take $\R$ as the following ordered set:
$$\begin{array}{rl}
\R  =  & \{r_1, \dots, r_{4^{\delta-1}} \} \\
    =  & \{0+ 2 \cdot 0,     \dots, \alpha^{\ell-1}+2 \cdot 0,
                  \dots,
   0+2 \cdot \alpha^{\ell-1} , \dots, \alpha^{\ell-1}+2\cdot      \alpha^{\ell-1}\}. \\
\end{array} $$
Since $|\mathcal R |/\delta=f_{0,\delta}+1$, we can form $f_{0, \delta}+1$ disjoints sets of $\mathcal R$ of size $\delta$.
For all $i \in \{0, \dots, f_{0,\delta}\}$, we consider the $\delta \times \delta$ quaternary matrix
\begin{equation*}
\mathcal N_i^*= \left(\begin{array}{c c}
1         & r_{\delta i+1} \\
\vdots    & \vdots \\
1         & r_{\delta (i+1)} \\
\end{array} \right).
\end{equation*}

\begin{theorem}
\label{construccionZ4}
Let $\mathcal P_{f_{0, \delta}}=\{\mathcal M_i \; : \; 0 \leq i \leq f_{0,\delta}\}$, where $\mathcal M_ i=\mathcal N_i^{-1}$. Then, $\Phi(\mathcal P_{f_{0, \delta}})$ is an $f_{0, \delta}$-${\PD}$-set of size $f_{0,\delta}+1$ for the $\mathbb Z_4$-linear Hadamard code $H_{0, \delta}$ of length $2^{2\delta-1}$.
\end{theorem}

\begin{proof}
We need to prove that $r_{\delta i+2}-r_{\delta i+1}, \dots, r_{\delta(i+1)} -r_{\delta i+1}$ are linearly independent over $\Z_4$, for all $i \in \{0, \dots, f_{0,\delta}\}$, to guarantee that $\mathcal N_i\in \PAut(\mathcal H_{0, \delta})$. Note that these vectors are not zero divisors \cite{Z4}.
Since $\alpha^{\ell}=1$, $\{r_{\delta i+2}-r_{\delta i+1}, \dots, r_{\delta(i+1)} -r_{\delta i+1}\}$ is one of the following three sets:
$$\begin{array}{ll}
L_1= & \{ 1, \dots, \alpha^{\delta-2}\}, \\
L_2= & \{\alpha^{k+1}-\alpha^k, \dots, \alpha^{k+\delta-1}-\alpha^k\}, {\rm for \ some} \; k \in \{0, \dots, \ell-1\}, \\
L_3= & \{\alpha^{k+1}-\alpha^k, \dots, \alpha^{\ell-1}-\alpha^k, -\alpha^k + 2(b_j -b_i), \alpha^{\ell}-\alpha^k + 2(b_j -b_i), \dots,  \\
 & \;\; \alpha^{k+\delta-2}-\alpha^k + 2(b_j -b_i)\}, {\rm for \ some} \; b_i, b_j \in \T \; {\rm and} \;   k \in \{0, \dots, \ell-1\}.
\end{array}$$
Elements in $L_1$ are clearly linearly independent over $\Z_4$. Now we prove that the same property is satisfied in $L_2$. Assume on the contrary that there are some $\lambda_i \neq 0$, $i \in \{1, \dots, \delta-1\},$ such that $\sum \lambda_i (\alpha^{k+i}-\alpha^k)=0$. If $\lambda_i \in \{1, 3\}$ for at least one $i\in \{1, \dots, \delta-1\}$, we get a contradiction. Indeed, if we take modulo 2 in the previous linear combination, we obtain that $\sum \bar{\lambda}_i (\bar{\alpha}^{k+i}-\bar{\alpha}^k)=0$, where $\bar{\lambda_i} \in \mathbb Z_2$ and at least one $\bar{\lambda_i} \neq 0$. This is a contradiction by Lemma~\ref{primitive}. 
On the other hand, if $\lambda_i\in \{0,2\}$ for all $i\in \{1, \dots, \delta-1\}$ and there is at least one $\lambda_i=2$, then $\sum 2\lambda'_i (\alpha^{k+i}-\alpha^k)=2[\sum \lambda'_i (\alpha^{k+i}-\alpha^k)]=0$, where $\lambda_i' \in \{0,1\}$ and at least one $\lambda_i'=1$. Hence, $\sum \lambda'_i (\alpha^{k+i}-\alpha^k)=2\lambda$ for some $\lambda \in \mathcal R$, that is, it is a zero divisor.
By taking modulo 2, we obtain a contradiction again by Lemma~\ref{primitive}.

We show that elements in $L_3=\{v_1, \dots, v_{\delta-1}\}$ are also linearly independent over $\Z_4$ by using a slight modification of the previous argument. Suppose that there is at least one $\lambda_i \neq 0$, $i \in \{1, \dots, \delta-1\},$ such that $\sum \lambda_i v_i=0$.  By taking modulo 2, we obtain that $\bar{\lambda}_{\delta-1} \bar{\alpha}^k+\sum \bar{\lambda}_i (\bar{\alpha}^{k+i}-\bar{\alpha}^k)= \bar{\alpha}^k[\bar\lambda_{\delta-1}+\sum \bar{\lambda}_i (\bar{\alpha}^{i}-1)]=0$. Since $\bar{\alpha}^k$ is a unit, it follows that $\bar\lambda_{\delta-1}+\sum \bar{\lambda}_i (\bar{\alpha}^{i}-1)=0$, which gives a contradiction if $\lambda_i \in \{1,3\}$ for at least one index, since $1, \bar{\alpha}-1, \dots, \bar{\alpha}^{\delta-2}-1$ are linearly independent over $\mathbb Z_2$. Note that the binary matrix
$$\left(\begin{array}{cc}
1               &  \mathbf{0} \\
\mathbf{1}      & \operatorname{Id}_{m-1}  \\
\end{array}\right),$$
has determinant 1. If $\lambda_i\in \{0,2\}$ for all $i \in \{1, \dots, \delta-1\}$, we get a contradiction by applying a similar argument to 
the one used above.

Finally, by construction, matrices $\mathcal N_i^*$ have no rows in common, since their rows are different elements of $\R$. By Theorem~\ref{principalZ4}, the result follows.
\end{proof}

\begin{example}
\label{4PDSET}
Let $\mathcal H_{0,3}$ be the quaternary linear Hadamard code of length 16 and type $2^04^3$. Let $\mathcal R = \mathbb Z_4[x]/(h(x))$, where $h(x)=x^2+x+1$. Note that $h(x)$ is a primitive basic irreducible polynomial dividing $x^3-1$ in $\mathbb Z_4[x]$. Let $\alpha$ be a root of $h(x)$. Then, $\T=\{0, 1, \alpha, \alpha^2\}$ and elements in $\R$ are ordered as follows:
$$\begin{array}{rl}
\R = & \{r_1, \dots, r_{16}\} \\
   = & \{0, 1, \alpha, 3+3\alpha, 2, 3, 2 + \alpha, 1+ 3\alpha,2\alpha, 1+2\alpha, \\
     & \; \;  3\alpha, 3+\alpha,2+2\alpha, 3+2\alpha, 2+3\alpha, 1+\alpha \}. \\
\end{array}$$
It is easy to check that matrices $\mathcal{N}_0^*=\operatorname{Id}_3^*$,
\begin{equation*}
\mathcal{N}_1^*= \left(\begin{array}{ccc}
1 & 3 & 3   \\
1 & 2 & 0    \\
1 & 3 & 0    \\
\end{array} \right) ,
\qquad
\mathcal{N}_2^*= \left(\begin{array}{ccc}
1 & 2 & 1   \\
1 & 1 & 3    \\
1 & 0 & 2    \\
\end{array} \right),
\qquad
\mathcal{N}_3^*= \left(\begin{array}{ccc}
1 & 1 & 2   \\
1 & 0 & 3    \\
1 & 3 & 1    \\
\end{array} \right),
\qquad
\mathcal{N}_4^*= \left(\begin{array}{ccc}
1 & 2 & 2   \\
1 & 3 & 2    \\
1 & 2 & 3    \\
\end{array} \right),
\end{equation*}
have no rows in common. Let $\mathcal P_4=\{\mathcal{N}_0^{-1}, \mathcal{N}_1^{-1}, \mathcal{N}_2^{-1}, 
\mathcal{N}_3^{-1}, \mathcal{N}_4^{-1} \}$, where
$\mathcal{N}_0=\operatorname{Id}_ 3$,
\begin{equation*}
\mathcal{N}_1= \left(\begin{array}{ccc}
1 & 3 & 3   \\
0 & 3 & 1    \\
0 & 0 & 1    \\
\end{array} \right) ,
\qquad
\mathcal{N}_2= \left(\begin{array}{ccc}
1 & 2 & 1   \\
0 & 3 & 2    \\
0 & 2 & 1    \\
\end{array} \right),
\qquad
\mathcal{N}_3= \left(\begin{array}{ccc}
1 & 1  & 2   \\
0 & 3  & 1    \\
0 & 2  & 3    \\
\end{array} \right),
\qquad
\mathcal{N}_4= \left(\begin{array}{ccc}
1 & 2  & 2   \\
0 & 1  & 0    \\
0 & 0  & 1    \\
\end{array} \right).
\end{equation*}
The set $\Phi(\mathcal P_4)$ is a 4-$\PD$-set of size 5 for $H_{0,3}$. Note that the bound $f_5=4$ is attained for $H_{0,3}$ despite the search of the 4-$\PD$-set is done in the subgroup $\Phi(\PAut(\mathcal H_{0,3}) )\leq \PAut(H_{0,3})$, since $f_{0,3}=f_5=4$.

\end{example}

We have that the $\Z_4$-linear Hadamard code of length $2^m$ with $\delta=1$ or $\delta=2$ is equivalent to the binary linear Hadamard code of length $2^m$ \cite{K}. However, the technique explained for binary linear Hadamard codes in Section \ref{sec:binarylinear} provides better results (in terms of $s$) that the one explained for  $\Z_4$-linear Hadamard codes when applied for linear codes, since $f_{\gamma, \delta}\leq f_m$, where $m=\gamma+2\delta-1$.

\begin{example}
We have provided a 2-$\PD$-set of size 3 for the binary linear Hadamard code $H_4$ of length 16 in Example \ref{ejemplo}. The code $H_4$ is equivalent to both $\Z_4$-linear Hadamard codes $H_{1,2}$ and $H_{3,1}$. However, a 2-$\PD$-set of size 3 is not achievable for $H_4$ by using Theorem~\ref{principalZ4}, since $f_{1,2}=f_{3,1}=1$.
\end{example}

\begin{example}
The binary linear Hadamard code $H_5$ of length 32 admits a 4-$\PD$-set of size 5 by Theorem~\ref{principal}, since $f_5=4$. Considering $H_5$ as the Gray map image of the quaternary linear Hadamard code $\mathcal H_{2,2}$ or $\mathcal H_{4,1}$, no more than a 3-$\PD$-set of size 4 can be found by using Theorem~\ref{principalZ4}, since $f_{4,1}=2$ and $f_{2,2}=3$.
\end{example}

\section{Recursive construction of $s$-$\PD$-sets for $\Z_4$-linear Hadamard codes}
\label{sec:Z4linearRecursive}

In this section,  given an $s$-$\PD$-set of size $l$ for the $\Z_4$-linear Hadamard code $H_{\gamma,\delta}$ of length $2^m$ and type $2^\gamma 4^\delta$, where $m=\gamma+2\delta-1$ and $l \geq s+1$, we show how to construct recursively an $s$-$\PD$-set of the same size for $H_{\gamma+i,\delta+j}$ of length $2^{m+i+2j}$ and type $2^{\gamma+i}4^{\delta+j}$ for all $i,j\geq 0$.

We first provide a recursive construction considering the elements of $\operatorname{PAut}(\mathcal H_{\gamma, \delta})$ as matrices in $\operatorname{GL}(\gamma+\delta, \mathbb Z_4)$. This construction can be seen as a natural generalization of the technique introduced for binary linear Hadamard codes in Section \ref{sec:binarylinearRecursive}.  Given a matrix $\mathcal M \in \PAut(\mathcal H_{\gamma, \delta})$ and an integer $\kappa\geq 1$, we define
\begin{equation}
\label{kappaconstructionZ4}
\mathcal M(\kappa)=\left(\begin{array}{cccc}
1 & \eta & \mathbf{0} & 2\theta \\
\mathbf{0} & A & \mathbf {0} & 2X \\
\mathbf{0} & \mathbf{0} & \operatorname{Id}_{\kappa} & \mathbf{0} \\
\mathbf{0} & \zeta(Y)& \mathbf{0} & \zeta(B) \\
\end{array}\right).
\end{equation}
\begin{proposition}
\label{recursivoZ4}
Let $\mathcal P_s=\{\mathcal{M}_0, \dots, \mathcal M_s\}\subseteq \operatorname{PAut}(\mathcal H_{\gamma, \delta})$ such that $\Phi(\mathcal P_s)$ is an $s$-$\PD$-set of size $s+1$ for $H_{\gamma, \delta}$ with information set $\Phi(\mathcal I_{\gamma, \delta})$. Then, $\mathcal Q_s=\{(\mathcal M^{-1}_0(\kappa))^{-1}, \dots, \mathcal M^{-1}_s(\kappa))^{-1} \} \subseteq \operatorname{PAut}(\mathcal H_{\gamma+i, \delta+j})$ and $\Phi(\mathcal Q_s)$ is an $s$-$\PD$-set of size $s+1$ for $H_{\gamma+i, \delta+j}$ with information set $\Phi(\mathcal I_{\gamma+i, \delta+j})$, for any $i, j \geq 0$ such that $i+j= \kappa \geq 1$.
\end{proposition}

\begin{proof}
Note that if $\mathcal M \in  \operatorname{PAut}(\mathcal H_{\gamma, \delta})$, construction (\ref{kappaconstructionZ4}) provides an element $\mathcal M(\kappa) \in  \operatorname{GL}(\gamma+\delta+ \kappa, \mathbb Z_4)$. Taking this into account, together with the fact that $\operatorname{Id}_{\kappa}$ can split as
$$\operatorname{Id}_{\kappa}=\left(\begin{array}{cc}
\operatorname{Id}_j & \mathbf{0} \\
\mathbf{0} & \operatorname{Id}_i \\
\end{array}\right),$$
where $i+j=\kappa\geq 1$, it is obvious that $\mathcal M^{-1}(\kappa)\in \operatorname{PAut}(\mathcal H_{\gamma+i, \delta+j})$ and so its inverse. Thus, $\mathcal Q_s\subseteq \operatorname{PAut}(\mathcal H_{\gamma+i, \delta+j})$. Finally, repeated rows in matrices $(\mathcal M^{-1}_0(\kappa))^*, \dots, (\mathcal M^{-1}_s(\kappa))^*$ cannot occur, since this fact implies repeated rows in matrices  $(\mathcal M^{-1}_0)^*, \dots, (\mathcal M^{-1}_s)^*$ by construction $(\ref{kappaconstructionZ4})$. The result follows from Theorem \ref{principalZ4}. 
\end{proof}

\begin{example}
Let $\mathcal P_4=\{\mathcal M_0, \dots, \mathcal M_4\} \subseteq \operatorname{PAut}(\mathcal H_{0,3})$ be the set, given in Example \ref{4PDSET},
such that $\Phi(\mathcal P_4)$ is a 4-$\PD$-set of size 5 for $H_{0,3}$. By Proposition \ref{recursivoZ4}, the set $\mathcal Q_4=\{\mathcal M^{-1}_i(1))^{-1} \; : \; 0 \leq i \leq 4\}$ is contained in both $\operatorname{PAut}(\mathcal H_{1,3})$ and $\operatorname{PAut}(\mathcal H_{0,4})$. Moreover, $\Phi(\mathcal Q_4)$ is a 4-$\PD$-set of size 5 for $H_{1, 3}$ and $H_{0, 4}$. Nevertheless, it is important to note that the construction of $(\mathcal M^{-1}_i(1))^{*}$ depends on the group where $\mathcal M^{-1}_i(1)$ is considered.
\end{example}

As for binary linear Hadamard codes, a second recursive construction considering the elements of $\operatorname{PAut}(H_{\gamma, \delta})$ as permutations of coordinate positions, that is as elements of $\Sym(2^m)$, can also be provided.
Given four permutations $\sigma_i \in \operatorname{Sym}(n_i), \; i \in \{1, \dots,  4\}$, we define  $(\sigma_1 | \sigma_2 | \sigma_3 | \sigma_4)\in \operatorname{Sym}(n_1+n_2+n_3+n_4)$ in the same way as we defined $(\sigma_1 | \sigma_2) \in \operatorname{Sym}(n_1+n_2)$ in Section \ref{sec:binarylinearRecursive}.

\begin{proposition}
\label{doublepdset}
Let $S$ be an $s$-$\PD$-set of size $l$ for $H_{\gamma, \delta}$ of length $n$ and type $2^{\gamma}4^{\delta}$ with information set $I$. Then, $( S |  S)=\{(\sigma | \sigma) : \sigma \in S\}$ is an $s$-$\PD$-set of size $l$ for $H_{\gamma+1, \delta}$ of length $2n$ and type $2^{\gamma+1}4^{\delta}$ constructed from (\ref{double}) and the Gray map, with any information set $I'=I \cup \{i+n\}$, $i\in I$.
\end{proposition}

\begin{proof}
Since $H_{\gamma+1, \delta}=\{(x, x), (x, \bar x) : x \in H_{\gamma, \delta}\}$, where $\bar x$ is the complementary vector of $x$,
the result follows using the same argument as in the proof of Proposition \ref{doublepdsetZ2}.
By the proof of Proposition \ref{infoset}, we can add any of the coordinate positions of $\{i+n : i \in I\}$ to $I$ in order to form a suitable information set $I'$ for $H_{\gamma+1, \delta}$. 
\end{proof}

Proposition \ref{doublepdset} cannot be generalized directly for $\Z_4$-linear Hadamard codes $H_{\gamma, \delta+1}$ constructed from (\ref{quadruple}) and the Gray map.
Note that if $S$ is an $s$-$\PD$-set for $H_{\gamma, \delta}$, then  $(S | S | S | S)=\{(\sigma | \sigma | \sigma | \sigma) : \sigma \in S\}$ is not always an $s$-$\PD$-set for $H_{\gamma, \delta+1}$, since in general $(\sigma | \sigma | \sigma | \sigma) \notin \operatorname{PAut}(H_{\gamma, \delta})$. For example,  $\sigma=(1, 5)(2, 8, 3, 6, 4, 7) \in \operatorname{PAut}(H_{0,2})\subseteq \Sym(8)$, but $\pi=(\sigma | \sigma | \sigma | \sigma) \notin \operatorname{PAut}(H_{0, 3})\subseteq\Sym(32)$, since $\pi(\Phi((0,0,0,0,1,1,1,1,2,2,2,2,3,3,3,3)))=\Phi((0,0,0,0,0,2,0,2,2,2,2,2,2,0,2,0))\not \in H_{0,3}$.

\begin{proposition}
\label{propquadruple}
Let $\mathcal{S}\subseteq \operatorname{PAut}(\mathcal H_{\gamma, \delta})$ such that
$\Phi(\mathcal S)$ is an $s$-$\PD$-set of size $l$ for $H_{\gamma, \delta}$ of length $n$ and type $2^{\gamma}4^{\delta}$ with information set $I$.  Then, $\Phi((\mathcal{S} | \mathcal{S} | \mathcal{S} | \mathcal{S}))=\{\Phi((\tau | \tau | \tau | \tau)) : \tau \in \mathcal S\}$ is an $s$-$\PD$-set of size $l$ for  $H_{\gamma, \delta+1}$ of length $4n$ and type $2^{\gamma}4^{\delta+1}$ constructed from (\ref{quadruple}) and the Gray map, with any information set $I''=I \cup \{i+n,j+n\}$, $i,j\in I$ and $i\not =j$.
\end{proposition}

\begin{proof}
Since $\mathcal H_{\gamma, \delta+1}$ is constructed from (\ref{quadruple}), $\mathcal H_{\gamma, \delta+1}=\{(u, u, u, u), (u, u + \mathbf{1}, u +\mathbf{2}, u+ \mathbf{3}), (u, u + \mathbf{2}, u, u+ \mathbf{2}), (u, u + \mathbf{3}, u, u+ \mathbf{1}) : u \in \mathcal H_{\gamma, \delta}\}$. It is easy to see that if $\tau \in  \operatorname{PAut}(\mathcal H_{\gamma, \delta})$, then $(\tau | \tau | \tau | \tau) \in  \operatorname{PAut}(\mathcal H_{\gamma, \delta+1})$.

Let $\sigma=\Phi(\tau)$.
Finally, we need to prove that for every $e \in \mathbb Z_2^{4n}$ with $\wt(e) \leq s$, there is $(\sigma | \sigma | \sigma | \sigma) \in \Phi((\mathcal{S} | \mathcal{S} | \mathcal{S} | \mathcal{S}))$ such that $(\sigma | \sigma | \sigma | \sigma)(e)_{I''} = \mathbf{0}$, where $I''\subseteq \{1, \dots, 4n\}$ is an information set for $H_{\gamma, \delta+1}$ with $\gamma+2 (\delta+1)$ coordinate positions. Using a similar argument to that given in the proofs of Propositions~\ref{doublepdsetZ2} and \ref{doublepdset},
the result follows.
Moreover, by the proof of Proposition~\ref{infoset}, any $I''=I \cup \{i+n,j+n\}$
with $i,j\in I$ and $i\not =j$ is a suitable information set for $H_{\gamma, \delta+1}$. 
\end{proof}


Propositions \ref{doublepdset} and \ref{propquadruple} can be applied recursively to acquire $s$-$\PD$-sets for any $\Z_4$-linear Hadamard codes obtained (by using constructions (\ref{double}) and (\ref{quadruple})) from a given $\Z_4$-linear Hadamard code where we already have such set. With this aim in mind, let denote by $2S$ the set $(S | S)$ and by $2^iS=2(2( \overset{i)}{\dots} (2S))$.

\begin{corollary}
\label{doublequadruple}
Let $\mathcal S\subseteq \operatorname{PAut}(\mathcal H_{\gamma, \delta})$ such that
$\Phi(\mathcal S)$ is an $s$-$\PD$-set of size $l$ for $H_{\gamma, \delta}$ of length $2^m$ and type $2^{\gamma}4^{\delta}$ with information set $I$. Then, $\Phi(2^{i+2j}\mathcal S)$ is an $s$-$\PD$-set of size $l$ for $H_{\gamma+i, \delta+j}$ of length $2^{m+i+2j}$ and type $2^{\gamma+i}4^{\delta+j}$ with information set obtained by applying recursively Proposition \ref{infoset}, for all $i, j\geq 0$.
\end{corollary}

\begin{proof}
The result comes trivially by applying Propositions \ref{infoset}, \ref{doublepdset} and \ref{propquadruple}. 
\end{proof}

%

Note that, from Theorem \ref{construccionZ4} and Proposition \ref{doublepdset}, we have explicitly provided an $f_{0, \delta}$-${\PD}$-set of size $f_{0, \delta}$+1 for each nonlinear $\mathbb Z_4$-linear Hadamard code $H_{\gamma, \delta}$, $\gamma \geq 0, \delta \geq 3$.

\section{Conclusions}

An alternative permutation decoding method that can be applied to $\mathbb Z_2 \mathbb Z_4$-linear codes \cite{Z2Z4Lineals}, which include $\mathbb Z_4$-linear codes, was presented in \cite{BeBoFeVi}. However, it remained as an open question to determine PD-sets for some families of $\mathbb Z_2\mathbb Z_4$-linear codes. In this paper, the problem of finding $s$-$\PD$-sets of minimum size $s+1$ for binary linear and $\mathbb Z_4$-linear Hadamard codes is addressed by finding $s+1$ invertible matrices over $\mathbb Z_2$ or $\mathbb Z_4$, respectively, which satisfy certain conditions.
Moreover, note that the first examples and constructions of (nonlinear) $\Z_4$-linear Hadamard codes are provided.  This approach establishes equivalent results to the ones obtained for simplex codes in \cite{FKM}.

As a future research in this topic, it would be interesting to provide explicitly an $f_{\gamma, \delta}$-${\PD}$-set of size $f_{\gamma, \delta}+1$ for $H_{\gamma, \delta}$, $s$-$\PD$-sets of minimum size for $\mathbb Z_2\mathbb Z_4$-linear Hadamard codes in general \cite{PRV}, or other families of $\Z_4$-linear codes such as Kerdock codes \cite{Z4}.


\begin{thebibliography}{16}

\bibitem{BV} R. Barrolleta and M. Villanueva,
\newblock ``Partial permutation decoding for binary linear Hadamard codes,"
{\em Electron. Note Discr. Math.} {\bf 46}, 35–42 (2014).

\bibitem{BeBoFeVi}
J. J. Bernal, J. Borges, C. Fern\'andez-C\'orboda, and M. Villanueva,
\newblock ``Permutation decoding of $\mathbb Z_2\mathbb Z_4$-linear codes,"
\newblock {\em Des. Codes and Cryptogr.} {\bf 76}(2), 269--277 (2015).

\bibitem{Z2Z4Lineals}
J. Borges, C. Fern\'{a}ndez-C\'{o}rdoba, J. Pujol, J. Rif\`{a}, and M. Villanueva,
\newblock ``$\Z_2\Z_4$-linear codes: generator matrices and duality,"
\newblock {\em Des. Codes and Cryptogr.} {\bf 54}, 167--179 (2010).

\bibitem{FKM}
W. Fish, J. D. Key, and E. Mwambene,
\newblock ``Partial permutation decoding for simplex codes,"
\newblock {\em Adv. Math. Commun.} {\bf 6}(4), 505--516 (2012).


\bibitem{G}
D. M. Gordon,
\newblock ``Minimal permutation sets for decoding the binary Golay codes,"
\newblock {\em IEEE Trans. Inf. Theory} {\bf 28}(3), 541--543 (1982).

\bibitem{H}
W. C.  Huffman,
{\em Codes and groups, Handbook of coding theory}, eds. V. S. Pless and W. C. Huffman, Elsevier (1998).

\bibitem{Z4}
A.~R. Hammons, Jr, P.~V. Kumar, A.~R. Calderbank, N.~J.~A. Sloane, and
  P.~Sol\'e, ``The {$Z_4$}-linearity of {K}erdock, {P}reparata, {G}oethals, and
  related codes,'' {\em IEEE Trans. Inf. Theory} {\bf 40}(2), 301--319, (1994).


\bibitem{KV2}
H.-J. Kroll and R. Vicenti,
\newblock ``PD-sets for binary RM-codes and the codes related to the Klein quadric and to the Schubert variety of PG(5,2),"
\newblock {\em Discrete Math.} {\bf 308}, 408--414, (2008).

\bibitem{KVi}
D. S. Krotov and M. Villanueva
\newblock 	``Classification of the $\Z_2\Z_4$-linear Hadamard codes and their automorphism groups,"
\newblock {\em IEEE Trans. Inf. Theory} {\bf 61}(2), 887--894 (2015).

\bibitem{K}
D. S. Krotov,
\newblock ``$\Z_4$-linear Hadamard and extended perfect codes,"
\newblock {\em Electron. Note Discr. Math.} {\bf 6}, 107-112 (2001).

\bibitem{M}
F. J. MacWilliams,
\newblock ``Permutation decoding of systematics codes,"
\newblock {\em Bell System Tech. J.} {\bf 43}, 485--505 (1964).

\bibitem{MSl}
F. J. MacWilliams and N. J. A. Sloane,
\newblock {\em The Theory of Error-Correcting Codes},
\newblock North-Holland Publishing Company (1977).

\bibitem{PRV}
K.T. Phelps,  J. Rif\`{a}, and M. Villanueva,
\newblock ``On the additive $\Z_4$-linear and non-$\Z_4$-linear Hadamard codes. Rank and Kernel,"
\newblock {\em IEEE Trans. on Information Theory} {\bf 52}(1), 316--319 (2005).

\bibitem{PPV}
J. Pernas, J. Pujol, and M. Villanueva,
\newblock ``Characterization of the automorphism group of quaternary linear Hadamard codes,"
\newblock {\em Des. Codes Cryptogr.} {\bf 70}(1-2), 105--115 (2014).

\bibitem{S}
P. Seneviratne,
\newblock ``Partial permutation decoding for the first-order Reed-Muller codes,"
\newblock {\em Discrete Math.} {\bf 309}(8), 1967--1970 (2009).

\bibitem{W}
J.  Wolfmann,
\newblock ``A permutation decoding of the (24,12,9) Golay code,"
\newblock {\em IEEE Trans. on Information Theory}  {\bf 29}(5), 748--750 (1983).



\end{thebibliography}
\end{document}